\documentclass[a4paper,USenglish]{lipics}

\usepackage[T1]{fontenc}
\usepackage[utf8]{inputenc}
\usepackage{lmodern}
\usepackage{microtype}
\usepackage{xspace,color,tikz,graphicx}
\usepackage[fleqn]{amsmath}
\usepackage{ellipsis}
\usepackage{cite}
\usepackage[hang]{footmisc}
\newtheorem{observation}{Observation}

\let\leftold\left
\let\rightold\right
\renewcommand{\left}{\mathopen{}\mathclose\bgroup\leftold}
\renewcommand{\right}{\aftergroup\egroup\rightold}

\newcommand{\pty}{\ensuremath{\pi}\xspace}
\newcommand{\copty}{\ensuremath{\pi^c}\xspace}
\newcommand{\cpty}{\ensuremath{\overline{\pi}}\xspace}
\newcommand{\forb}{\ensuremath{S_\pty}\xspace}
\newcommand{\obli}{\ensuremath{S_{\bar{\pi}}}\xspace}
\newcommand{\alg}{\ensuremath{\textsc{Alg}}\xspace}
\newcommand{\algOpt}{\ensuremath{\textsc{Opt}}\xspace}
\newcommand{\tape}{\ensuremath{\phi}\xspace}
\newcommand{\acc}{\ensuremath{S_\textsc{Alg}}\xspace}

\newcommand{\maxP}{\ensuremath{\textsc{Max-}\scalebox{1.5}{\pty}}\xspace}
\newcommand{\minP}{\ensuremath{\textsc{Min-}\scalebox{1.5}{\pty}}\xspace}
\newcommand{\maxASG}{\ensuremath{\textsc{maxASGk}}\xspace}
\newcommand{\maxASGu}{\ensuremath{\textsc{maxASGu}}\xspace}
\newcommand{\co}[1]{\ensuremath{\overline{#1}}\xspace}
\newcommand{\cost}[1]{\ensuremath{\textrm{cost}(#1)}\xspace}
\newcommand{\profit}[1]{\ensuremath{\textrm{profit}(#1)}\xspace}
\newcommand\ie{i.\,e.}

\title{Advice Complexity of the Online Induced Subgraph Problem\footnote{Supported in part by the Villum Foundation and the Stibo-Foundation and SNF grant 200021-146372.}}

\author[1]{Dennis Komm}
\author[2]{Rastislav Kr\'alovi\v{c}}
\author[3]{Richard Kr\'alovi\v{c}}
\author[4]{Christian Kudahl}

\affil[1]{Dept.\ of Computer Science, ETH Zurich; \texttt{dennis.komm@inf.ethz.ch}}
\affil[2]{Dept.\ of Computer Science, Comenius University; \texttt{kralovic@dcs.fmph.uniba.sk}}
\affil[3]{Google Inc., Switzerland; \texttt{richard.kralovic@dcs.fmph.uniba.sk}}
\affil[4]{Dept.\ of Mathematics and Computer Science, University of Southern Denmark\\\texttt{kudahl@imada.sdu.dk}}

\begin{document}

\maketitle

\begin{abstract}
  Several well-studied graph problems aim to select a largest (or smallest) induced
  subgraph with a given property of the input graph.  Examples of such problems include maximum
  independent set, maximum planar graph, maximum induced clique, maximum
  acyclic subgraph (a.k.a.\ minimum feedback vertex set), and many others.  In
  online versions of these problems, vertices of the graph are presented in an
  adversarial order, and with each vertex, the online algorithm must irreversibly
  decide whether to include it into the constructed subgraph, based only on the
  subgraph induced by the vertices presented so far.  We study the properties
  that are common to all these problems by investigating the generalized
  problem: for an arbitrary but fixed hereditary property \pty, find some maximal
  induced subgraph having \pty.  We study this problem from the point of view of
  advice complexity, \ie, we ask how some additional information about the yet
  unrevealed parts of the input can influence the solution quality.  We evaluate
  the information in a quantitative way by considering the best possible
  advice of given size that describes the unknown input.  Using a result from Boyar
  et al.\ [STACS 2015, LIPIcs 30], we give a tight
  trade-off relationship stating that for inputs of length $n$ roughly $n/c$ bits of advice are both needed and
  sufficient to obtain a solution with competitive ratio $c$, regardless of the
  choice of \pty, for any $c$ (possibly a function of $n$).  This complements the
  results from Bartal et al.\ [SIAM Journal on Computing 36(2), 2006] stating that,
  without any advice, even a randomized
  algorithm cannot achieve a competitive ratio better than $\Omega(n^{1-\log_{4}3-o(1)})$.

  Surprisingly, a similar result cannot be obtained for the symmetric problem:
  for a given cohereditary property \pty, find a minimum subgraph having \pty. We
  show that the advice complexity of this problem varies significantly with
  the choice of \pty.
  
  We also consider the so-called preemptive online model, inspired by some
  application mainly in networking and scheduling, where the decision of the
  algorithm is not completely irreversible.  In particular, the algorithm may
  discard some vertices previously assigned to the constructed set, but
  discarded vertices cannot be reinserted into the set again.  We show that, for
  the maximum induced subgraph problem, preemption cannot help much, giving a
  lower bound of $\Omega(n/(c^2\log c))$ bits of advice needed to obtain
  competitive ratio $c$, where $c$ is any increasing function bounded by
  $\sqrt{n/\log n}$.  We also give a linear lower bound for $c$ close to $1$.
\end{abstract}

\section{Introduction}\label{sec:intro}

Online algorithms get their input gradually, and this way have to produce parts
of the output without full knowledge of the instance at hand, which is a large
disadvantage compared to classical \emph{offline computation}, yet a realistic
model of many real-world scenarios \cite{BE98}.  Most of the offline problems
have their online counterpart.  Instead of asking about the time and space
complexity of algorithms to solve a computational problem, \emph{competitive
analysis} is commonly used as a tool to study how well online algorithms
perform  \cite{BE98,ST85} without any time or space restrictions;  the
analogous offline measurement is the analysis of the \emph{approximation
ratio}.  A large class of computational problems for both online and offline
computation are formulated on graphs; we call such problems (online)
\emph{graph problems}.

In this paper, we deal with problems on unweighted undirected graphs that are
given to an online algorithm vertex by vertex in consecutive discrete time steps.
Formally, we are given a graph $G=(V,E)$, where $|V|=n$, with an ordering
$\prec$ on $V$.  Without loss of generality, assume $V=\{v_1,\dots,v_n\}$, and
$v_1\prec v_2 \prec\dots\prec v_n$ specifies the order in which the vertices of
$G$ are presented to an online algorithm; this way, the vertex $v_i$ is given
in the $i$th time step.  Together with $v_i$, all edges $\{v_j,v_i\}\in E$ are
revealed for all $v_j\prec v_i$.  If $v_i$ is revealed, an online algorithm
must decide whether to accept $v_i$ or discard it.  Neither $G$ nor $n$ are
known to the online algorithm.  We study two versions of online problems; with
and without preemption.  In the former case, the decision whether $v_i$ is
accepted or not is definite.  In the latter case, in every time step, the
online algorithm may preempt (discard) some of the vertices it previously
accepted; however, a vertex that was once discarded cannot be part of the
solution anymore.

For an instance $I=(v_1,\dots,v_n)$ of some graph problem, we denote by
$\alg(I)$ the solution computed by some online algorithm \alg; $\algOpt(I)$ denotes an
optimal solution for $I$, which can generally only be computed with the full
knowledge of $I$.  We assume that $I$ is constructed in an \emph{adversarial
manner} to give worst-case bounds on the solution quality of any online
algorithm.  This means that we explicitly think of $I$ as being given by an
adversary that knows \alg and wants to make it perform as poorly as possible;
for more details, we refer to the standard literature \cite{BE98}.

For maximization problems with an associated \emph{profit function}
$\textrm{profit}$, an online algorithm \alg is called \emph{$c$-competitive}
if, for every instance $I$ of the given problem, it holds that
\begin{equation}\label{eq:cr1}
  \profit{\alg(I)} \ge 1/c\cdot \profit{\algOpt(I)}\;;
\end{equation}
likewise, for minimization problems with a \emph{cost function} $\textrm{cost}$, we require
\begin{equation}\label{eq:cr2}
 \cost{\alg(I)} \le c\cdot \cost{\algOpt(I)}
\end{equation}
for every instance $I$.
In this context, $c$ may be a positive constant or a function that increases
with the input length $n$. We will use $c$ and $c(n)$ interchangeably to refer to the competitive ratio;
the latter is simply used to emphasize that $c$ may depend on $n$. 

Throughout this paper, $\log$ denotes the binary logarithm $\log_2$.
For $q\in\mathbb{N}$, let $[q]=\{0,1,\ldots,q-1\}$.

Instead of studying specific graph problems, in this paper, we investigate a large
class of such problems, which are defined by \emph{hereditary properties}.
This class includes many well-known problems such as \emph{maximum independent set},
\emph{maximum planar graph}, \emph{maximum induced clique}, and
\emph{maximum acyclic subgraph}.
We call any collection of graphs a \emph{graph property} \pty.
A graph has property $\pty$ if it is in the collection.
We only consider properties
that are \emph{non-trivial}, \ie,  they are  both true for infinitely many
graphs and false for infinitely many graphs.
A property  is called
\emph{hereditary} if it holds that, if a graph $G$ satisfies \pty, then also any
induced subgraph $G'$ of $G$ satisfies \pty; conversely, it is called
\emph{cohereditary} if it holds that, if a graph $G$ satisfies \pty, and $G$ is
an induced subgraph of $G'$, then also $G'$ satisfies \pty.  For a graph
$G=(V,E)$ and a subset of vertices $S=\{v_1,\dots, v_i\}\subseteq V$, let
$G[S]$ (or $G[v_1,\dots,v_i]$) denote the subgraph of $G$ induced by the
vertices from $S$.  For a graph $G=(V,E)$, let $\co{G}=(V,\co{E})$ be the
complement of $G$, \ie, $\{u,v\}\in\co{E}$ if and only if $\{u,v\}\not\in E$.
Let $K_n$ denote the complete graph on $n$ vertices, and let $\co{K}_n$ denote the
independent set on $n$ vertices.
We consider the online version of the problem of finding maximal (minimal,
respectively) induced subgraphs satisfying a hereditary (cohereditary,
respectively) property \pty, denoted by \maxP (\minP, respectively).
For the ease of presentation, we will call such problems \emph{hereditary (cohereditary, respectively) problems}.
Let $\acc:=\alg(I)$ denote the set of vertices accepted by some online algorithm \alg
for some instance $I$ of a hereditary problem.  Then, for \maxP, the
profit of \alg is $|\acc|:=\profit{\alg(I)}$ if $G[\acc]$ has the property \pty and
$-\infty$ otherwise; the goal is to maximize the profit.  Conversely, for \minP, the
cost of \alg is $|\acc|:=\cost{\alg(I)}$ if $G[\acc]$ has the property \pty and
$\infty$ otherwise; the goal is to minimize the cost.
As an example, consider the \emph{online maximum independent set} problem;
the set of all independent sets is clearly a hereditary property (every independent set is a feasible solution, and
every induced subset of an independent set is again an independent set).  With every
vertex revealed, an online algorithm needs to decide whether it becomes part of
the solution or not.  The goal is to compute an independent set that is as
large as possible; the profit of the solution thus equals $|\acc|$.  It is
straightforward to define the problem without or with preemption.

In this paper, we study \emph{online algorithms with advice} for hereditary and
cohereditary problems.  In this setup, an online algorithm is equipped
with an additional resource that contains information about the instance it is
dealing with.  A related model was originally introduced by Dobrev et al.\
\cite{DKP09}.  Revised versions were introduced by Emek et al.\ \cite{EFKR11},
B\"ockenhauer et al.\ \cite{BKKKM09}, and Hromkovi\v{c} et al.\ \cite{HKK10}.
Here, we use the model of the latter two papers.  Consider an input $I=(v_1,\dots,v_n)$ of a
hereditary problem.  An \emph{online algorithm \alg with advice} computes the
output sequence $\alg^{\tape}(I)=(y_1,\dots,y_n)$ such that $y_i$ is computed
from $\tape,v_1,\dots,v_i$, where \tape is the content of the advice tape, \ie,
an infinite binary sequence.  We denote the cost (profit, respectively) of the computed output by
$\cost{\alg^{\tape}(I)}$ ($\profit{\alg^{\tape}(I)}$, respectively).  The algorithm \alg is \emph{$c$-competitive with
advice complexity $b(n)$}, if for every $n$ and for each $I$ of length at most
$n$, there exists some \tape such that $\cost{\alg^{\tape}(I)}\le c\cdot
\cost{\algOpt(I)}$ ($\profit{\alg^{\tape}(I)}\ge (1/c)\cdot \profit{\algOpt(I)}$,
respectively) and at most the first $b(n)$ bits of \tape have been accessed
by \alg.\footnote{Note that usually an additive constant is included in the definition
of $c$-competitiveness, \ie, in \eqref{eq:cr1} and \eqref{eq:cr2}.  However, for the
problems we consider, this changes the advice complexity by at most $O(\log n)$;
see Remark 9 in Boyar et al.\ \cite{BFKM15}).}

The motivation for online algorithms with advice is mostly of theoretical
nature, as we may think of the information necessary and sufficient to compute
an optimal solution as the \emph{information content} of the given problem
\cite{HKK10}.  Moreover, there is a non-trivial connection to randomized online
algorithms \cite{BKKK11,KK11}.  Lower bounds from advice complexity often
translate to lower bounds for semi-online algorithms.  One could consider if
knowing some small parameter for an online problem (such as the length of the
input or the number of requests of a certain type) could result in a much
better competitive ratio.  Lower bound results using advice can often help
answer this question.  Similarly, lookahead can be seen as a special kind of
advice that is supplied to an algorithm.  This way, online algorithms with advice
generalize a number of concepts introduced to give online algorithms more
power.  However, the main question posed is, how much could any kind of (computable)
information help; and maybe even more importantly, which amount of information will
never help to overcome some certain threshold, no matter what this information
actually is.

\subsection*{Organization, Related Work, and Results}

We are mainly concerned with proving lower bounds of the form that a particular
number of advice bits is necessary in order to obtain some certain output quality
for a given hereditary property.  We make heavy use of online reductions between
generic problems and the studied ones that allow to bound the number of advice
bits necessary from below.  Emek et al.\ \cite{EFKR11} used this technique in
order to prove lower bounds for the $k$-server problem.  The foundations of the
reductions as we perform them here are due to B\"ockenhauer et al.\ \cite{BHKKSS14},
who introduced the \emph{string guessing problem}, and Boyar et al.\ \cite{BFKM15},
who studied a problem called \emph{asymmetric string guessing}.
Mikkelsen \cite{mikkelsen15} introduced a problem, which we call the \emph{anti-string guessing problem},
and which is a variant of string guessing with a more ``friendly'' cost function.
Our reductions rely on some results from Bartal et al.\ \cite{BFL06} that
characterize hereditary properties by forbidden subgraphs together with some
insights from Ramsey theory.

In Section \ref{sec:prelim}, we recall some basic results from Ramsey theory, and define
the generic online problems that we use as a basis of our reductions.  In
Section \ref{sec:no_preem}, we study both \maxP and \minP in the case that no
preemption is allowed;  using a reduction from the asymmetric string guessing
problem, we show that any $c$-competitive online algorithm for \maxP needs
roughly $n/c$ advice bits, and this is essentially tight.  This complements the
results from Bartal et al.\ \cite{BFL06} stating that, without any advice, even
a randomized algorithm cannot achieve a competitive ratio better than
$\Omega(n^{1-\log_{4}3-o(1)})$.  Note that the advice complexity of the maximum
independent set problem on bipartite and sparse
graphs was studied by Dobrev et al.\ \cite{DRR15}.  In the subsequent sections,
we allow the online algorithm to use preemption.  First, in
Section \ref{sec:preem_large}, we use a reduction from the string guessing problem to
show a lower bound of $\Omega(n/(c^2\log c))$ bits of advice that are needed to obtain
competitive ratio $c$, where $c$ is any increasing function bounded by
$\sqrt{n/\log n}$.   In Section \ref{sec:preem_small}, using a reduction from the anti-string
guessing problem, we also give a linear lower bound for $c$ close to $1$.

\section{Preliminaries}\label{sec:prelim}

Hereditary properties can be characterized by forbidden induced subgraphs as
follows: if a graph $G$ does not satisfy a hereditary property \pty, then any
graph $H$ such that $G$ is an induced subgraph of $H$ does not satisfy \pty
neither.  Hence, there is a (potentially infinite) set of minimal forbidden graphs
(w.r.t.\ being induced subgraph) \forb such that $G$ satisfies \pty if and only if no
graphs from \forb are induced subgraphs of $G$.  Conversely, any set of graphs
$S$ defines a hereditary property $\pty_S$ of not having a graph from $S$ as
induced subgraph.

Furthermore, there is the following  bijection between hereditary and cohereditary
properties: for a hereditary property \pty we can define a property \cpty such
that a graph $G$ satisfies \cpty if and only if it does not satisfy \pty (it is
easy to see that \cpty is cohereditary), and vice versa.  Hence, a cohereditary
property \cpty can be characterized by a set of minimal (w.r.t.\ being induced subgraph) \emph{obligatory} subgraphs \obli such
that a graph $G$ has the property \cpty if and only if at least one graph from
\obli is an induced subgraph of $G$.

To each property \pty we can define the complementary property \copty such that
a graph $G$ satisfies \copty if and only if the complement of $G$ satisfies
\pty. Clearly, if \pty is (co)hereditary, so is \copty.  Moreover, if $H$ is
forbidden (obligatory, respectively) for \pty, $\overline{H}$ is forbidden (obligatory, respectively) for
\copty.  The following statement is due to Lund and Yannakakis.

\begin{lemma}[Lund and Yannakakis \cite{LY80}, proof of Theorem 4]\label{lm:pty-is-or-cliques}
  Every non-trivial hereditary property \pty is satisfied either by all cliques or by all independent sets.
\end{lemma}

\begin{proof}
  Assume, for the sake of contradiction, that there is a hereditary property
  \pty, and two numbers $m$, $n$, such that $K_m$ and $\co{K}_n$ do not satisfy
  \pty. Let $r(m,n)$ be the Ramsey number\cite{R30}, such that every graph with
  at least $r(m,n)$ vertices contains $K_m$ or $\co{K}_n$ as induced subgraph.
  Since \pty is non-trivial, there is a graph $G$ with more than $r(m,n)$
  vertices that satisfies \pty. $G$ contains either $K_m$ or $\co{K}_n$ as
  induced subgraph, and since \pty is hereditary, either $K_m$ of $\co{K}_n$
  satisfies \pty.
\end{proof}

Bartal et al.\ proved the following theorem.

\begin{theorem}[Bartal et al.\ \cite{BFL06} and references therein]
  In the known supergraph mod\-el, any randomized algorithm for the \maxP problem has competitive ratio
  $\Omega(n^{1-\log_43-o(1)})$, even if preemption is allowed.
  \qed
\end{theorem}

The theorem is formulated in the known supergraph model, where a graph
$G=(V,E)$ with $n$ vertices is a-priori known to the algorithm, and the input
is a sequence of vertices $v_1,\dots,v_k$. The task is to select in an online
manner the subgraph of the induced graph $G[v_1,\dots,v_k]$, having property
\pty.  Note that $n$ in the previous theorem thus refers to the size of the known
supergraph, and not to the length of the input sequence. However, in the proof a
graph with $n=4^i$ vertices is considered, from which subgraphs of size $3^i$
are presented.  Each of these instances has an optimal solution of size at least
$2^i$, and it is shown that any deterministic algorithm can have a profit of at
most $\alpha(3/2)^i\log n$ on average, for some constant $\alpha$.  From that, using
Yao's principle \cite{Yao77} as stated in \cite{BE99}, the result follows.  The same set of
instances thus yields the following result.

\begin{theorem}[Bartal et al.\ \cite{BFL06}]\label{thm:maxP-random}
  Any randomized algorithm for the \maxP problem has competitive ratio
  $\Omega(n^{2/\log 3-1-o(1)})$, even if preemption is allowed.
  \qed
\end{theorem}

Next, we describe some specific online problems that allow us to give lower bounds on
the advice complexity using a special kind of reduction.

B\"ockenhauer et al.\ \cite{BHKKSS14} introduced a very generic online problem called
\emph{string guessing with known history over alphabets of size $\sigma$}
($\sigma$-SGKH).  The input is a sequence of requests $(x_0,\dots,x_n)$
where $x_0=n$ and for $i\ge 1$, $x_i\in\{1,\dots,\sigma\}$.
The algorithm has to produce a sequence of answers $(y_1,\dots,y_n,y_{n+1})$,
where $y_i\in\{1,\dots,\sigma\}$ and $y_{n+1}={\perp}$
and where $y_i$ is allowed to depend on $x_0,\dots,x_{i-1}$ (and of course any advice bits the algorithm reads).
The cost is the number of positions $i$ for which $y_i\ne x_{i}$.

\begin{theorem}[B\"ockenhauer et al.\ \cite{BHKKSS14}]\label{thm:SGKH}
  Any online algorithm with advice for $\sigma$-SGKH that guesses $\gamma n$ bits of the
  input correctly must read at least
  \[ \left(1+(1-\gamma)\log_{\sigma}\left(\frac{1-\gamma}{\sigma-1}\right)+\gamma\log_{\sigma}\gamma\right) n\log_2\sigma \]
  advice bits.
  \qed
\end{theorem}

Mikkelsen \cite{mikkelsen15} introduced the problem \emph{anti-string guessing with
known history over alphabets of size $\sigma$} (anti-$\sigma$-SGKH).  It is
defined exactly as $\sigma$-SGKH except that the cost is the number of
positions $i$ for which $y_i=x_i$.

\begin{theorem}[Mikkelsen {\cite[Theorem 11]{mikkelsen15}}]\label{thm:anti-lb}
  Let $\sigma\geq 2$ and let $1 \leq c < \frac{\sigma}{\sigma-1}$.
  A $c$-competitive anti-$\sigma$-SGKH algorithm must read at least
  \[ b\geq(1-h_\sigma(1/c))n\log_2 \sigma \]
  bits of advice, where $n$ is the input length. This holds even if $n$ is known in advance.
  Here, $h_\sigma$ is the $\sigma$-ary entropy function given by 
  $h_\sigma(x)=x\log_{\sigma}(\sigma-1)-x\log_\sigma(x)-(1-x)\log_\sigma(1-x)$.
  \qed
\end{theorem}

Boyar et al.\ \cite{BFKM15} investigated a problem called
\emph{maximum asymmetric string guessing} (\maxASG).  The input is a sequence of
requests $(x_0,x_1,\dots,x_n)$ where $x_0={\perp}$ and for $i\ge1$,
$x_i\in\{0,1\}$.  The algorithm has to produce a sequence of answers
$(y_1,\dots,y_n,y_{n+1})$.  The output is feasible if $x_i\le y_i$ for all
$1\le i\le n$.  The profit of the algorithm is the number of zeroes in
$y_1,\dots,y_n$ for feasible outputs, and $-\infty$ otherwise.  The ``blind''
version of the problem, where the algorithm has to produce the outputs without
actually seeing the requests (\ie, in each step, the algorithm receives some
dummy request $\perp$), is denoted \maxASGu.  In what follows, let
\[ B_c := \log\left(1+\frac{(c-1)^{c-1}}{c^c}\right) \approx\frac{1}{c}\cdot\frac{1}{\mathrm{e}\ln2}\;. \]

\begin{theorem}[Boyar et al.\ \cite{BFKM15}] \label{thm:maxasg-ub}
  For any function $c(n)$ such that $1\le c(n)\le n$, there is a $c$-competitive algorithm for \maxASG
  (\maxASGu, respectively) with advice of size $B_c\cdot n + O(\log n)$.
  Moreover, any $c$-competitive algorithm for \maxASG (\maxASGu, respectively) 
  must read at least $B_c\cdot n - O(\log n)$ bits of advice.
  \qed
\end{theorem}

Note that in general, it does not make much difference if the length of the input is initially known to the algorithm or not.
More specifically, it changes the advice complexity by at most $O(\log n)$.

\section[MaxPi and MinPi without Preemption]{\boldmath\maxP and \minP\unboldmath without Preemption}\label{sec:no_preem}

First, we show that for any non-trivial hereditary property \pty, the \maxP
problem is equivalent to the asymmetric string guessing in the following sense.

\begin{theorem}\label{thm:max-ub}
  If there is a  $c$-competitive algorithm for \maxASGu then there is a
  $c$-competitive algorithm for \maxP using the same advice.
\end{theorem}

\begin{proof}
  Suppose there is a $c$-competitive algorithm \alg for \maxASGu with
  $c(n)<\infty$.  For any input for \maxP, consider a binary string representing
  this instance (with one bit per vertex) such that zeroes
  correspond to vertices in the optimal solution.  Any feasible solution to
  \maxASGu must cover all ones in the input string, so the set of zeroes in any
  feasible solution of \maxASGu forms subset of vertices of the optimal
  solution of \maxP, and since \pty is hereditary, is a feasible solution of
  \maxP.

  The algorithm for \maxP works as follows: When a vertex arrives, send a
  $\perp$ request to \alg.  If \alg answers $0$, include the vertex in the
  solution.  Since \alg is $c$-competitive, it covers at least a
  $(1/c)$-fraction of the optimum.
\end{proof}

\begin{theorem}\label{thm:max-lb}
  If there is a $c$-competitive algorithm for \maxP that reads $b$ bits of
  advice then there is a $c$-competitive algorithm for \maxASG using
  $b+O(\log^2 n)$ bits of advice.
\end{theorem}

Before proving Theorem \ref{thm:max-lb} let us recall Lemma~3 from Bartal et al.\ \cite{BFL06}.

\begin{lemma}[Bartal et al.\ \cite{BFL06}]\label{lm:ramseyish}
  Given any graph $H$, there exist constants $n_0$ and $\alpha$ such that for
  all $n > n_0$ there exists a graph $G$ on $n$ vertices such that any induced
  subgraph of $G$ on at least $\alpha \log n$ vertices contains $H$ as an
  induced subgraph.
\end{lemma}

This is a variant of Lemma~9 from Lund and Yannakakis \cite{LY93}.

\begin{lemma}[Lund and Yannakakis \cite{LY93}]\label{lm:ramsey2}
  Let $H$ be a graph on $k$ vertices.  For sufficiently large $N$, for any graph
  $G$ on $N$ vertices and for all $\ell=\Omega(\log N)$, a pseudo-random subgraph
  $G'$ of $G$ does not, with probability $1/2$, contain a subset $S$ of $\ell$
  vertices that is a clique in $G$ but $H$ is not an induced subgraph of $G'|S$.
\end{lemma}

\begin{proof}[Proof of Theorem \ref{thm:max-lb}]
  According to Lemma \ref{lm:pty-is-or-cliques}, \pty is satisfied either by all cliques or by all independent sets.
  Without loss of generality, suppose the latter (otherwise, swap the edges and
  non-edges in the following arguments).

  Consider a binary string $\nu=x_1,\dots,x_n$ (for large enough $n$). 
  Let us consider the graph $G_\nu=(V,E)$ defined as follows.
  Let $H$ be an arbitrary but fixed forbidden subgraph of \pty.  Let $G'$ be the $n$-vertex graph
  from Lemma \ref{lm:ramseyish} with vertices $V=\{v_1,\dots,v_n\}$.  If $x_i=0$ for some $i$, delete from $G'$
  all edges $\{v_i,v_j\}$ for $j>i$. In the graph $G_\nu$ thus defined, the vertices $v_i$ for which the
  corresponding $x_i$ satisfies $x_i=0$ (denoted by $I_\nu\subseteq V$ in the sequel)
  form an independent set, and hence $G_\nu[I_\nu]$ has property \pty.  On the 
  other hand, any induced subgraph $G_\nu[S]$ with property \pty
  can contain at most $\alpha\log n$ vertices from $V\setminus I_\nu$ (otherwise it would contain the forbidden
  graph $H$ as induced subgraph).  Note that, with $O(\log n)$ bits of advice to encode $n$,
  the graph $G_\nu$ can be constructed from the string $\nu$ in an online manner:
  the base graph $G'$ is fixed for a fixed $n$, and the subgraph $G_\nu[v_1,\dots,v_i]$
  depends only on the values of $x_1,\dots,x_{i-1}$.

  Now let us consider a $c$-competitive algorithm $\alg_\pty$ for \maxP that uses $b$ bits of advice. 
  Let us describe how to derive an algorithm \alg for \maxASG from $\alg_\pty$.  For a given string $\nu=x_1,\dots,x_n$, where 
  $\perp,x_1,\dots,x_n$ is the input for \maxASG, the advice for \alg consists of three parts: first, there
  is a self-delimited encoding of $n$ using $O(\log n)$ bits, followed by a (self-delimited)
  correction string $e_\nu$ 
  of length $O(\log^2 n)$ bits described later, and the rest is
  the advice for $\alg_\pty$ on the input $G_\nu$.
  Let $S$ be the solution (set of vertices) returned by $\alg_\pty$ on $G_\nu$ (with the proper advice).
  As argued before, $S$ can contain at most $\alpha\log n$ vertices from $V\setminus I_\nu$. The
  indices of these vertices from $S_{\text{out}}:=S\cap(V\setminus I_\nu)$ are part of the string $e_\nu$.
  Apart from that, $e_\nu$ contains the indices of at most $\alpha\log n$ vertices $S_{\text{in}}\subseteq I_\nu$
  such that $|(S\setminus S_{\text{out}})\cup S_{\text{in}}|=\min\{|S|,|I_\nu|\}$.
  
  The algorithm \alg works as follows: at the beginning, it constructs the graph $G'$. When a request 
  $x_i$ arrives, \alg sends the new vertex $v_i$ of $G_\nu$ to $\alg_\pty$, and finds out whether $v_i\in S$.
  If $v_i\in S_{\text{in}}$, \alg answers $0$ regardless of the answer of $\alg_\pty$. Similarly, if $v_i\in S_{\text{out}}$,
  \alg answers $1$.  Otherwise, \alg answers $0$ if and only if $v_i\in S$.
  
  First note that \alg always produces a feasible solution: if the input $x_i=1$
  then either $v_i\not\in S$, and \alg returns $y_i=1$, or else $v_i$ is
  included in $S_{\text{out}}$.  Moreover, the number of zeroes (the profit)
  in the output of \alg
  is $\min\{|S|,|I_\nu|\}$, where $|I_\nu|$ is the profit of the optimal solution. Since $\alg_\pty$ is
  $c$-competitive, $|S|\ge(1/c)\cdot \profit{\algOpt(G_\nu)}\ge(1/c)\cdot|I_\nu|$.
\end{proof}

\begin{corollary} \label{cor:max-uplow}
  Let \pty be any non-trivial hereditary property. Let $A_{c,n}$ be the minimum advice needed for a  
  $c$-competitive \maxP algorithm.  Then
  \[ B_c\cdot n - O(\log n) \le A_{c,n} \le B_c\cdot n + O(\log ^2n)\;. \]
\end{corollary}

We have shown that the advice complexity of \maxP essentially does not depend
on the choice of the property \pty. This is not the case with cohereditary
properties and the problem \minP.  On one hand, there are cohereditary
properties where little advice is sufficient for optimality:

\begin{theorem} \label{thm:min-low}
  If a cohereditary property $\pty$ can be characterized by finitely many obligatory subgraphs, there
  is an optimal algorithm for \minP with advice $O(\log n)$.
\end{theorem}

\begin{proof}
  Since each obligatory subgraph has constant size, $O(\log n)$ bits can be used to encode the indices of the 
  vertices (forming the smallest obligatory subgraph) that are included in an optimal solution.
\end{proof}

On the other hand, there are properties for which the problem \minP requires
large advice.

\begin{theorem}[Boyar et al.\ \cite{BFKM15}] \label{thm:min-high}
  Any $c$-competitive algorithm for the \emph{minimum cycle finding} problem requires at least
  $B_c\cdot n - O(\log n)$ bits of advice.
  \qed
\end{theorem}

The problem \emph{minimum cycle finding} requires to identify a smallest
possible set of vertices $S$ such that $G[S]$ contains a cycle. Hence, it is
the \minP problem for the non-trivial cohereditary property ``contains cycle.\!''
An upper bound analogous to Theorem \ref{thm:max-ub} follows from \cite{BFKM15}.

\begin{theorem}\label{cohereditary-alg} \label{thm:min-up}
  Let \pty be any non-trivial cohereditary property.  There is a $c$-competitive algorithm for \minP
  which reads $B_c\cdot n+O(\log n)$ bits of advice.
  \qed
\end{theorem}

\section[MaxPi with Preemption -- Large Competitive Ratios]{\boldmath\maxP\unboldmath with Preemption -- Large Competitive Ratios}\label{sec:preem_large}

In this, and the following section, we consider the problem \maxP with
preemption where \pty is a non-trivial hereditary property.  In every time
step, the algorithm can either accept or reject the currently given vertex and
preempt any number of vertices that it accepted in previous time steps.
However, vertices that were once rejected or preempted cannot be accepted in
later time steps.  The goal is to accept as many vertices as possible. After
reach request, the solution is required to have the property \pty.\footnote{Note
that without advice, the condition to maintain \pty in every time step is
implicit.  Indeed, if \pty  should be violated in any step, the adversary may
just end the input sequence.  Then again, allowing advice changes this game.
Here, an online algorithm may just accept all vertices, compute the optimal
solution at the end, and preempt the corresponding vertices.  Therefore, we
require explicitly that \pty must be maintained in every time step.}  Using
a string guessing reduction, we prove the following theorem.

\begin{theorem}\label{thm:lower_maxp_1}
  Consider the \maxP problem with preemption, for a hereditary property \pty with a forbidden subgraph $H$,
  such that \pty holds for all independent sets.
  Let $c(n)$ be an increasing function such that $c(n)\log c(n)=o(\sqrt{n/\log n})$. 
  Any $c(n)$-competitive \maxP algorithm 
  uses at least $\Omega\left(\frac{n}{c(n)^2\log c(n)}\right)$ bits of advice.
\end{theorem}

\begin{proof}
  First, for some given $n$ and $\sigma$, let us define the graph $G_{n,\sigma}$ that will be used in the reduction. 
  To ease the presentation, assume that $n'=n/\sigma$ is integer. 
  Let $G_1$ be a graph with $\sigma$ vertices,
  the existence of which is asserted by Lemma \ref{lm:ramseyish}, such that any subgraph of $G_1$
  with at least $\kappa_1\log\sigma$ vertices contains $H$ as induced subgraph.
  Let $G_B$ be the complement of a union of $n'$ cliques of size $\sigma$ (\ie, $G_B$ consists of $n'$ independent 
  sets $V_1,\dots,V_{n'}$ of size $\sigma$, and all remaining pairs of vertices are connected by edges).
  Applying Lemma \ref{lm:ramsey2} to $G_B$ proves an existence of a graph $G_2\subseteq G_B$ such that any subset
  of $G_2$ with at least $\kappa_2\log n$ vertices contains $H$ as an induced subgraph.
  The graph $G_{n,\sigma}$ is obtained from $G_2$ by replacing each independent set 
  $V_i$ with a copy of $G_1$ (each such copy is called a ``layer'' in what follows).

  Let us suppose that a $c(n)$-competitive \maxP algorithm \alg is given 
  that uses $S(n)$ advice bits on instances of size
  $n$. Now fix an arbitrary $n$, and choose $\sigma:=4c\kappa_1\log(4c\kappa_1)$.
  We show how to solve instances of $\sigma$-SGKH of length $n'-1$ using \alg.
  Let $q_1,\dots,q_{n'-1}$ be the instance of $\sigma$-SGKH, where $q_i\in\{1,\dots,\sigma\}$.
  The corresponding instance $G$ for the \maxP problem is as follows: take the graph $G_{n,\sigma}$, and denote
  by $v_{i,1},\dots,v_{i,\sigma}$ the vertices of the set $V_i$.  Let $v_{i,q_i}$ be the distinguished 
  vertex in set $V_i$.  Delete from $G_{n,\sigma}$ all edges of the form $\{v_{i,q_i},v_{i',q_{i'}}\}$ where $i'>i$.
  The resulting graph $G$ is presented to \alg in the order $v_{1,1},\dots,v_{1,\sigma},v_{2,1},\dots,v_{2,\sigma},\dots$.
  
  Note that $G$ can be constructed online based on the instance $q_1,\dots,q_{n'-1}$.

  The distinguished vertices form an independent set of size $n'$, and thus a feasible solution. On the other hand,
  apart from the distinguished vertices, any solution can have at most $\kappa_1\log\sigma$ vertices in one layer
  (otherwise there would be a forbidden subgraph in that layer), and at most $\kappa_2\log n$ layers with vertices 
  other than the distinguished one (if 
  there are more than $\kappa_2\log n$ nonempty layers, choose one vertex from each nonempty layer;  these form a 
  clique in $G_B$, and due to Lemma \ref{lm:ramsey2} induce $H$ in $G_2$, and thus also in $G$).
  Hence, $n'\le \profit{\algOpt(G)}\le n'+K$, where $K:=\kappa_1\kappa_2\log\sigma\log n$.

  Since \alg is $c$-competitive, it produces a solution of size at least $\profit{\algOpt(G)}/c$. Since any solution can have at 
  most $K$ non-distinguished vertices, the solution of \alg contains at least $g:=\profit{\algOpt(G)}/c-K$ distinguished vertices.

  Consider an algorithm $\alg'$ for $\sigma$-SGKH of length $n'-1$, which simulates \alg: for the $i$th request, it 
  presents \alg with  the layer of vertices $V_i$. Let $\textrm{Cand}(i)\subseteq V_i$ (the \emph{candidate} set)
  be the set of 
  vertices selected by \alg 
  from $V_i$. As stated before, $|\textrm{Cand}(i)|\le\kappa_1\log\sigma$.  A set $\textrm{Cand}(i)$ is \emph{good} if it contains the 
  distinguished vertex $v_{i,q_i}$.  It follows from the definition of the problem that there are at least $g$ good 
  candidate sets.

  $\alg'$ uses an additional $O(\log\log\sigma)$ bits of advice to describe a number $j$, 
  $1\le j\le\kappa_1\log\sigma$, and selects the $j$th vertex from any $\textrm{Cand}(i)$ as an answer (if $|\textrm{Cand}(i)|$ is
  smaller than $j$, it is extended in an arbitrary fixed way). The number $j$ is selected in such a way that
  $\alg'$ gives the correct answer from a fraction of $1/(\kappa_1\log\sigma)$ of good sets. 
  Putting it together, the fraction of correctly guessed numbers by $\alg'$ is at least
  \[ \alpha:=\frac{n'-cK}{c\kappa_1\log\sigma(n'-1)}. \]
  Note that
  \[ \frac{1}{c\kappa_1\log\sigma}\ge\alpha\ge\frac{1}{2c\kappa_1\log\sigma} \]
  holds for large enough $n$, provided that $n'\ge2cK-1$. To see that this inequality holds, note that
  \[ n'\ge2cK-1 \iff \frac{n}{4c\kappa_1\log(4c\kappa_1)}\ge2cK-1\iff(2cK-1)4c\kappa_1\log(4c\kappa_1)\le n\;.\]
  The last inequality holds for large enough $n$ by the choice of $c(\cdot)$ due to the fact that
  \[ (2cK-1)4c\kappa_1\log(4c\kappa_1) = O(c(n)^2K\log c(n)) = O((c(n)\log c(n))^2\log n) = o(n)\;. \]
  
  Due to Theorem \ref{thm:SGKH}, any algorithm for $\sigma$-SGKH that correctly guesses a fraction
  of $\alpha$ numbers (for $\frac{1}{\sigma}\le\alpha\le1$) on an input of length $n'-1$ requires at least 
  $S:=F(\sigma,\alpha)\cdot(n'-1)\cdot\log_2\sigma$ bits of advice where
  \[ F(\sigma,\alpha):=1+(1-\alpha) \log_\sigma\left(\frac{1-\alpha}{\sigma-1}\right)+\alpha \log_\sigma\alpha\;. \]
  First, let us verify that $1/\sigma\le\alpha$:
  \begin{alignat*}{3}
        &&\frac{1}{\sigma}\le\frac{1}{2c\kappa_1\log\sigma} &\le \alpha \\
   \iff && \frac{\sigma}{\log\sigma}                    &\ge 2c\kappa_1 \\
   \iff && \frac{4c\kappa_1\log(4c\kappa_1)}{\log(4c\kappa_1)+\log\log(4c\kappa_1)} &\ge 2c\kappa_1 \\
   \iff && 2\log(4c\kappa_1) &\ge \log(4c\kappa_1) + \log\log(4c\kappa_1) \\
   \iff && \log(4c\kappa_1)  &\ge \log\log(4c\kappa_1)\;. 
  \end{alignat*}
  Next, we use the bounds on $\alpha$ to get
  \begin{align*}
    F(\sigma,\alpha) &\ge 1 + \left(1-\frac{1}{c\kappa_1\log\sigma}\right)\cdot \frac{\log\left(\frac{\left(1-\frac{1}{c\kappa_1\log\sigma}\right)}{\sigma-1}\right)}{\log\sigma} + \frac{1}{2c\kappa_1\log\sigma}\cdot\frac{\log\left(\frac{1}{2c\kappa_1\log\sigma}\right)}{\log\sigma} \\
                     &\ge 1 - \left(1-\frac{1}{c\kappa_1\log\sigma}\right)\cdot\frac{\log(\sigma-1)}{\log\sigma} - \frac{1}{\log\sigma} \left[ \frac{\log\left(\frac{c\kappa_1\log\sigma}{c\kappa_1\log\sigma-1}\right)} {\frac{c\kappa_1\log\sigma}{c\kappa_1\log\sigma-1}} + \frac{\log(2c\kappa_1\log\sigma)}{2c\kappa_1\log\sigma}\right] \\
                     &\ge \frac{1}{c\kappa_1\log\sigma}-\frac{1}{\log\sigma}\left[ \frac{1}{\ln2(c\kappa_1\log\sigma-1)}+\frac{\log(2c\kappa_1\log\sigma)}{2c\kappa_1\log\sigma} \right]
  \end{align*}
  and hence
  \[ F(\sigma,\alpha)\log\sigma\ge
     \frac{1}{c\kappa_1}- \frac{1}{\ln2(c\kappa_1\log\sigma-1)} -
     \frac{\log(2c\kappa_1\log\sigma)}{2c\kappa_1\log\sigma}\;. \]
  Since $\sigma\approx c\log c$, it holds 
  \[ \frac{1}{\ln2(c\kappa_1\log\sigma-1)} = o\left(\frac{1}{c}\right) \]
  and for $n\mapsto\infty$
  \[ \frac{\log(2c\kappa_1\log\sigma)}{2c\kappa_1\log\sigma} \mapsto \frac{1}{2c\kappa_1} \]
  yielding $F(\sigma,\alpha)\log\sigma = \Omega(1/c)$.
  Finally, the theorem follows by noting that $n'-1 = \Omega(n/(c\log c))$.
\end{proof}

Using a similar approach, we can get a stronger bound for the independent
set problem.  Due to space constraints, the proof is moved to Appendix \ref{app:ind_preem}.

\begin{theorem}\label{thm:ind_preem}
  Let $c(n)$ be any function such that 
  \[ 8\le c(n)\le\frac{1+\sqrt{1+4n}}{4}\;. \]
  Any $c(n)$-competitive independent set algorithm that can use preemption
  requires at least
  \[ A_{c,n}\ge\frac{0.01\cdot\log(2c)}{2c^2}(n-2c) \]
  advice bits.
  \qed
\end{theorem}

\section[MaxPi with Preemption -- Small Competitive Ratios]{\boldmath\maxP\unboldmath with Preemption -- Small Competitive Ratios}\label{sec:preem_small}

In this section, we use Theorem \ref{thm:anti-lb} to give bounds on small constant
values of the competitive ratio for algorithms for \maxP complementing the
bounds from Theorem \ref{thm:lower_maxp_1}.  In what follows, \pty is a non-trivial
hereditary property and $k$ is the size of a smallest forbidden subgraph
according to \pty.

\begin{theorem}\label{thm:lower_maxp_2}
  If there is a $c$-competitive algorithm for \maxP with preemption that
  reads $b(nk)$ bits of advice for inputs of length $nk$, then there exists a
  $c$-competitive algorithm for Anti-$k$-SGKH,
  which, for inputs of length $n$, reads $b(kn)+O(\log^2 n)$ bits of advice.
\end{theorem}

\begin{proof}
  According to Lemma \ref{lm:pty-is-or-cliques} \pty is satisfied either by
  all cliques or by all independent sets.  We assume in the following, that
  \pty is satisfied by all independent sets.  We describe how to transform an
  instance of Anti-$k$-SGKH into an instance of
  \maxP with preemption.  The length of the instance for \maxP with preemption
  will be $k$ times as long as the length $n$ of the Anti-$k$-SGKH instance.
  We proceed to show that a $c$-competitive
  algorithm for the latter implies a $c$-competitive algorithm for the former
  which reads at most $O((\log n)^2)$ additional advice bits.
  
  Let $\nu=x_1,x_2,\dots,x_n$ (with $x_i \in \{1,\dots,k$) be an instance of Anti-$k$-SGKH.  Consider the $n$-vertex graph
  $\tilde{G}=(V(\tilde{G}),E(\tilde{G}))$, given by Lemma \ref{lm:ramseyish} for
  a size-$k$ smallest minimal forbidden subgraph $H=(V(H),E(H))$ from \pty.  We
  recall that any induced subgraph of $\tilde{G}$ with at least $\alpha \log n$
  vertices contains $H$ as an induced subgraph.  Let us denote
  $V(\tilde{G})=\{\tilde{v}_1,\dots,\tilde{v}_n\}$ and $V(H)=\{h_1,\dots,h_k\}$.
  We now describe the construction of a graph $G_\nu=(V(G_\nu),E(G_\nu))$,
  which will be the input for the given algorithm for \maxP.  To this end, let
  \begin{align*}
    V(G_\nu) &= \bigcup_{i=1}^{n} \bigcup_{j=1}^{k}v^i_j, \\
    E(G_\nu) &= \{\{v^i_j,v^i_{j'}\} \mid \{h_j,h_{j'}\}\in E(H)\} \cup 
                \{\{v^i_j,v^{i'}_{j'}\} \mid i < i', \{\tilde{v_j},\tilde{v_{j'}}\} \in E(\tilde{G}), j \neq x_i\},
  \end{align*}
  where we assume an ordering $v^1_1, \dots, v^1_k, v^2_1, \dots, v^2_k, \dots, v^n_1, \dots, v^n_k$
  on the vertices.
  Moreover, we denote the requests $v^i_1, \dots, v^i_k$ as layer $i$.
  Let $X$ denote the set of vertices $v^i_{x_i}$ for $i =1, \dots, n$.
  
  We start with a few observations about $G_\nu$ that are straightforward.

  \begin{observation}
    The graph $G_\nu[X]$ is an independent set of size $n$.  In particular, it has property \pty.
  \end{observation}
  
  \begin{observation} \label{reduction2}
    $G_\nu[v^i_1, \dots, v^i_k]=H$ for an arbitrary but fixed $i$.
    Thus, any induced subgraph of $G_\nu$ that contains $G_\nu[v^i_1,
    \dots, v^i_k]$ not does not have property \pty.
  \end{observation}
  
  \begin{observation} \label{reduction3}
    Consider a set of vertices, $V$, in $G_\nu$ which is disjoint from $X$. If
    $|V| \geq k \alpha \log n$, then $G_\nu[V]$ does not have property \pty.
  \end{observation}

  Note that $V$ must in this case contain vertices from at least $\alpha \log
  n$ different layer.  These have $H$ as an induced subgraph since none of them
  are in $X$.
  
  Now consider a $c$-competitive algorithm $\alg_\pty$ for \maxP with
  preemption reading $b(nk)$ bits of advice (recall that $nk$ is the length of its input
  $G_\nu$).
  We start by describing an algorithm $\alg'$ for Anti-$k$-SGKH, which uses $b(nk)$ bits of advice ($n$ is the length of its input).
  Afterwards, we use it to define another algorithm $\alg$ for Anti-$k$-SGKH, which uses $O(\log^2 n)$ additional advice bits and is
  $c$-competitive.
  
  For a given string $\nu=x_1,\dots,x_n$, let $\perp,x_1,\dots,x_n$ be the
  input for Anti-$k$-SG.  Let $S$ be the solution (set of vertices)
  returned by $\alg_\pty$ on $G_\nu$ (with the proper advice).  Note that this
  is the resulting set of vertices after the unwanted vertices have been
  preempted.  $\alg'$ works as follows: It constructs the graph
  $G_\nu$ online and simulates $\alg_\pty$ on it.  When a request $i$ arrives,
  the goal of $\alg'$ is to guess a number in $\{1,\dots,k\}$ different from $x_i$.  It
  does this by presenting all vertices in layer $i$ to $\alg_\pty$.  It is
  important to note that the vertices in layer $i$ can be presented without
  knowledge of $x_i, \dots, x_n$.  Let $S_i$ denote the set of these vertices,
  which are accepted by $\alg_\pty$ and have not been preempted after request $v^i_k$.  In
  layer $i$, $\alg'$ outputs $y_i=w$ where $w$ is the smallest number in $\{1,\dots,k\}$
  such that $v^i_w \notin S_i$.  Note that such a number always exists because
  of Observation \ref{reduction2}.
  
  We now describe $\alg$, which uses $O(\log^2 n)$ additional advice bits.  The
  advice for $\alg$ consists of three parts: First, a self-delimiting encoding
  of $n$ (this requires $O(\log n)$ bits).  This is followed by a list of up to
  $k\alpha \log n$ indices $i$, where $\alg$ outputs $y_i=x_i$.  Let
  $S_{\text{error}}$ denote the set of these indices. A
  self-delimiting encoding of this requires $O( \log^2 n)$ bits (recall that
  $\alpha$ and $k$ are constant).  Finally, the advice which $\alg'$ received
  is included. This is $b(nk)$ bits.
  
  The algorithm \alg works as follows:
  For each request, it does the following:
  If the request is not in $S_{\text{error}}$ it outputs the same as $\alg$.
  If the request is in $S_{\text{error}}$ it outputs another number in $[k]$.
  
  We now argue that $\alg$ is $c$-competitive.  We note that the optimal
  offline solution to \maxP on $G_\nu$ contains at most $k\alpha \log n$
  vertices not in $X$.  The same of course holds for the solution produced by
  $\alg_\pty$.  It holds that if in layer $i$, $\alg_\pty$ accepts a vertex in
  $X$, then $\alg'$ outputs $y_i \neq x_i$.  This means that the score of
  $\alg_\pty$ is at most $k\alpha \log n$ more than the score of $\alg'$.
  Since the score of \alg is $k\alpha \log n$ more than the score of $\alg'$,
  we have that \alg is $c$-competitive.
\end{proof}

Combining Theorems \ref{thm:anti-lb} and \ref{thm:lower_maxp_2}, we get the following corollary.

\begin{corollary}
  For a non-trivial hereditary property, $\pty$, with a minimal forbidden subgraph of size $k$,
  the following holds:
  Let $1 < c < \frac{k}{k-1}$.
  A $c$-competitive algorithm for $\maxP$ with preemption must read at least
  \[ b\geq(1-h_k(1/c))n\frac{\log k}{k} -O(\log^2 n) \]
  bits of advice, where $n$ is the input length.
  Here, $h_k$ is the $k$-ary entropy function given by 
  $h_k(x)=x\log_{k}(k-1)-x\log_k(x)-(1-x)\log_k(1-x)$.
\end{corollary}

\section{Closing Remarks}\label{sec:closrem}

In Corollary \ref{cor:max-uplow}, we describe lower and upper bounds for the
advice complexity of all online hereditary graph problems, which are
essentially tight (there is just a gap of $O(\log^2 n)$).  It turns out that,
for all of them, roughly the same amount of information about the future is
required to achieve a certain competitive ratio.

Intriguingly, we see a quite different picture for cohereditary properties.
Theorem \ref{thm:min-up} gives the same upper bound as we had for hereditary
properties, and Theorem \ref{thm:min-high} shows that this upper bound is essentially
tight.  However, Theorem \ref{thm:min-low} shows that there exist cohereditary
problems that have an advice complexity as low as $O(\log n)$ bits to be
optimal.  It remains open if it is only those problems with a finite set of
obligatory graphs that have this very low advice complexity, or if this can also
happen for cohereditary problems with an infinite set of obligatory graphs.

For hereditary problems with preemption, we show that to achieve a competitive
ratio strictly smaller than $\frac{k}{k-1}$, a linear number of advice bits is
needed.  This is asymptotically tight, since optimality (even without
preemption) can be achieved with $n$ bits.  Furthermore, we show a lower
bound for non-constant competitive ratios (that are roughly smaller than
$\sqrt{n}$).  It remains open if there is an algorithm for the preemptive case,
which uses fewer advice bits than the algorithms solving the same problem in
the non-preemptive case.

\newpage

\begin{appendix}

\section{Proof of Theorem \ref{thm:ind_preem}}\label{app:ind_preem}

Consider an instance of $\sigma$-SGKH of length $n'-1$.  We construct an
unweighted undirected graph $G$ as input for the maximum independent set
problem with preemption as follows.  We subdivide the input into $n'=n/\sigma$
layers.  In every layer, vertex after vertex, $\sigma$ vertices are given that
form a clique.  These cliques are denoted by $C_1,\dots,C_{n'}$; the clique
$C_i$, $1\le i\le n'-1$, contains a designated vertex $c_i$ such that all
vertices that are revealed subsequently are connected to all vertices $v\in
C_i$ if and only if $v\ne c_i$.  This way, $G$ contains an independent set
(namely $\{c_1,\dots,c_{n'-1},v\}$ for any $v\in C_{n'}$) of size $n'$.

Note that due to the condition that an independent set is maintained in every
time step, no online algorithm is allowed to accept more than one vertex in any
layer.  If, for clique $C_i$ with $1\le i\le n'-1$, the algorithm accepts a
vertex $v\ne c_i$, in the next layer it has two choices.  It can preempt $v$
which reduces its profit by $1$, or it can keep $v$ and not have any more
profit on the rest of the instance. Clearly, the second option is never
superior to the first one, so we can assume, without loss of generality, that the algorithm
always preempts incorrectly guessed vertices.  Also, assuming that the
algorithm always selects exactly one vertex means no loss of generality (it has
no benefit not to select any vertex in a given layer).

Given a $c(n)$-competitive independent set algorithm $\alg$, one can construct
a $\sigma$-SGKH algorithm $\alg'$ as follows: for each request, $\alg'$
presents $\alg$ with a new set of $\sigma$ vertices, and the vertex selected by
$\alg$ is interpreted as the answer to $\sigma$-SGKH. The next request of
$\sigma$-SGKH reveals the correct guess, so the instance for $\alg$ can be
constructed incrementally.  The choice of the designated vertex $c_i$ for $C_i$
with $1\le i\le n'-1$ by \alg increases the profit of $\alg$ by 1, and
corresponds to the correct guess of the natural number from $1$ to $\sigma$ by
$\alg'$.  A wrong choice gets preempted, and does not contribute to the profit of
\alg.  Let 
\begin{equation}
  \alpha:=\frac{n'-c(n)}{c(n)n'-c(n)}\;.
\end{equation}
Since \alg is $c(n)$-competitive, and the optimum has size $n'$, \alg selects
an independent set of size at least $n'/c(n) = \alpha(n'-1)+1$, meaning that
the fraction of correct guesses was at least $\alpha$.  On the other hand, due
to Theorem \ref{thm:SGKH}, any algorithm for $\sigma$-SGKH that
correctly guesses a fraction of $\alpha$ numbers (for
$\frac{1}{\sigma}\le\alpha\le1$) on an input of length $n'-1$ requires at least
$S:=F(\sigma,\alpha)\cdot(n'-1)\cdot\log_2\sigma$ bits of advice where
\[ F(\sigma,\alpha):=1+(1-\alpha) \log_\sigma\left(\frac{1-\alpha}{\sigma-1}\right)+\alpha \log_\sigma\alpha\;. \]
We have 
\begin{equation}\label{eqn-alpha}
  \frac{1}{c}\ge\alpha=\frac{n'-c}{cn'-c}\ge\frac{1}{2c}\;.
\end{equation}
provided that $1\le c<\frac{n'+1}{2}$.
Using \eqref{eqn-alpha}, we get
\begin{align*}
  F(\sigma,\alpha) &\ge 1+\frac{c-1}{c}\cdot\frac{\log\left(\frac{c-1}{c(\sigma-1)}\right)}{\log\sigma}+\frac{1}{2c}\cdot\frac{\log\left(\frac{1}{2c}\right)}{\log\sigma} \\
                   &= 1-\frac{c-1}{c}\cdot\frac{\log(\sigma-1)}{\log\sigma}-\frac{1}{\log\sigma}\cdot\left[\frac{\log\left(\frac{c}{c-1}\right)}{\frac{c}{c-1}}+\frac{\log(2c)}{2c}\right] \\
                   &\ge\frac{1}{c}-\frac{1}{\log\sigma}\cdot\left[\frac{\log\left(\frac{c}{c-1}\right)}{\frac{c}{c-1}}+\frac{\log(2c)}{2c}\right]\;.
\end{align*}
Since for any $x\ge0$ it holds\footnote{Let $f(x)=\ln(1+x)$, $g(x)=x(1+x)$.  It
holds $f(0)=g(0)=0$. The derivatives are $f'(x)=\frac{1}{1+x}$, $g'(x)=1+2x$.
Both $f(x)$, $g(x)$  are increasing, $f'(0)=g'(0)=1$, and $f'(x)$ is decreasing
while $g'(x)$ is increasing.} $\ln(1+x)\le x(1+x)$, we get 
\[ F(\sigma,\alpha)\ge\frac{1}{c}-\frac{1}{\log\sigma}\cdot\left[\frac{1}{\ln2(c-1)}+\frac{\log(2c)}{2c}\right]
   \ge \frac{1}{c}-\frac{1}{\log\sigma}\cdot \frac{\log\left(4(2c)^{\ln2}\right)}{2\ln2(c-1)}\;. \]
Subsequently, for the required advice we get
\begin{equation}
  S \ge \frac{n-\sigma}{\sigma}\cdot\log\sigma\cdot F(\sigma,\alpha)
    \ge \frac{n-\sigma}{\sigma}\cdot\left[\frac{\log\sigma}{c}-\frac{\log\left(4(2c)^{\ln2}\right)}{2\ln2(c-1)}\right] \;.
\end{equation}
We have to choose $\sigma$ in such a way that $\alpha\ge\frac{1}{\sigma}$ and $c<\frac{\frac{n}{\sigma}+1}{2}$.
It is easy to see that for $\sigma=2c$ both inequalities hold for the assumed values of $c$; the second inequality
being equivalent to 
$4c^2-2c-n<0$.
The function 
\[ h(c):=\frac{\log\left(4(2c)^{\ln2}\right)}{2\ln2(c-1)}=\frac{2+\ln2\log(2c)}{2\ln2(c-1)} \]
converges to $\frac{1}{2}\frac{\log(2c)}{c}$. Hence, for sufficiently large $c$, 
$h(c)$ comprises a linear fraction of $\frac{\log\sigma}{c}$. Solving numerically, for $c\ge8$,
$h(c)\le0.99\cdot\frac{\log(2c)}{c}$,
yielding
\[ S\ge 0.01\cdot\frac{n-2c}{2c}\cdot\frac{\log(2c)}{c}\;, \]
which finishes the proof.
\end{appendix}

\end{document}